\documentclass{birkmult}
\usepackage{amsrefs}

 \newtheorem{thm}{Theorem}[section]

 \theoremstyle{definition}
 \newtheorem{defn}[thm]{Definition}
 \theoremstyle{remark}
 \newtheorem{rem}[thm]{Remark}
 \newtheorem{ex}[thm]{Example}
 \numberwithin{equation}{section}

\def \pdtu{\frac{\partial u}{\partial t}}

\begin{document}
%
%
%
%
%
%
%
%
%
\title[Stochastic Incompleteness]
 {Stochastically Incomplete Manifolds and Graphs}
\author[Rados\l aw K. Wojciechowski]{Rados{\l}aw Krzysztof Wojciechowski}

\address{%
Grupo de F\'isica-Matem\'atica \\
Complexo Interdisciplinar da Universidade de Lisboa \\
Av. Prof. Gama Pinto, 2 \\
PT--1649--003 Lisboa, Portugal}

\email{radoslaw@cii.fc.ul.pt}

\thanks{This work was completed with the support of FCT grant SFRH/BPD/45419/2008 and FCT project PTDC/MAT/101007/2008.}

\subjclass{Primary 39A12; Secondary 58J65}

\keywords{Stochastic incompleteness, explosion, heat kernel, manifolds, graphs, curvature, volume growth}

\date{October 1, 2009}


\begin{abstract}
We survey geometric properties which imply the stochastic incompleteness of the minimal diffusion process associated to the Laplacian on manifolds and graphs.  In particular, we completely characterize stochastic incompleteness for spherically symmetric graphs and show that, in contrast to the case of Riemannian manifolds, there exist examples of stochastically incomplete graphs of polynomial volume growth.
\end{abstract}

\maketitle
\section{Introduction}

A diffusion process whose lifetime is almost surely infinite is said to be  \emph{stochastically complete} (or \emph{conservative} or \emph{non-explosive}).  If this fails to occur, that is, if the total probability of the particle undergoing the diffusion to be in the state space is less than one at some time, the process is said to be \emph{stochastically incomplete}.   A trivial way for stochastic incompleteness to occur is to impose a killing boundary condition.  It is the objective of this article to survey the geometric properties, in the case when no such killing condition is present, that cause stochastic incompleteness to occur for the minimal diffusion process associated to the Laplacian on manifolds and graphs.  We draw heavily from the survey article of A. Grigor{\cprime}yan \cite{Gri99}  for the case of Riemannian manifolds and then present some recent results for graphs.  

Examples of geodesically complete but stochastically incomplete manifolds were first given by R. Azencott \cite{Az74}.  Specifically, if $M$ is a geodesically complete, simply connected, negatively curved, analytic Riemannian manifold and $k(r)$ denotes the smallest, in absolute value, sectional curvature at distance $r$, then Azencott showed that $M$ is stochastically incomplete if $k(r) \geq C r^{2+ \epsilon}$ for $\epsilon > 0$ \cite{Az74}*{Proposition 7.9}.  In these examples, the large negative curvature pushes the particle to infinity in a finite time and explosion occurs.  Furthermore, Azencott showed that such a manifold is stochastically complete if the sectional curvature is uniformly bounded from below.  In 1978, S.T. Yau showed that a geodesically complete manifold whose Ricci curvature is bounded from below is stochastically complete \cite{Yau78}.  A different type of criterion for stochastic completeness in terms of the volume growth was given by Grigor{\cprime}yan in 1986.  In particular, Grigor{\cprime}yan's result implies that, if $V(r) \leq e^{cr^2}$, where $V(r)$ denotes the volume of a geodesic ball of radius $r$, then a geodesically complete manifold is stochastically complete \cite{Gri86}.   The examples of Azencott, or the case of model manifolds, show that Grigor{\cprime}yan's criterion is sharp.  

The corresponding question for graphs was explicitly addressed by J. Dodziuk and V. Mathai in 2006.  By analyzing bounded solutions of the heat equation, they show that graphs of bounded valence are stochastically complete \cite{DodMat06}*{Theorem 2.10}.  Examples of a stochastically incomplete graph were given in \cites{Woj07, Woj09}.  These examples are trees branching rapidly in all directions from a fixed vertex.  More specifically, letting $k_+(r)$ denote the minimum number of outward pointing edges, where the minimum is taken over all vertices on a sphere of radius $r$ in a tree, then the diffusion is stochastically incomplete if  $\sum_{r=0}^\infty \frac{1}{k_+(r)} < \infty$ \cite{Woj09}*{Theorem 3.4}.  Therefore, in the case of graphs, the number of outward pointing edges plays the role of the negative sectional curvature in sweeping the particle out to infinity. The volume growth for such trees is factorial and, while smaller, at least comparable to the examples of Azencott.  In this article, we give many more examples of stochastically incomplete graphs.  In particular, in Theorem \ref{sgraphs}, we completely characterize the stochastic incompleteness of spherically symmetric graphs and use this to give examples of stochastically incomplete graphs with only polynomial volume growth.  Thus, in the case of graphs, no direct analogue of Grigor{\cprime}yan's theorem holds.

\section{Stochastic Incompleteness}
In this section we give an overview of properties equivalent to stochastic incompleteness.  Here, the manifold and graph settings are quite analogous so we do not distinguish between the two.  To avoid trivial examples, we assume that all manifolds are geodesically complete and that all graphs are infinite.  Furthermore, we assume that all underlying spaces are connected. 

We start by outlining the construction of the heat kernel.  In both cases, one has a Laplacian acting on a dense subset of the space of $L^2$ functions.  We choose our sign convention so that the Laplacian is a positive operator.  Therefore, in the case of $\mathbb{R}^n$, we take $\Delta = - \sum_{i=1}^n \frac{\partial^2}{\partial x_i^2}$, while, for graphs, if $f$ is a function on the vertices, then, pointwise, the Laplacian acts by $\Delta f(x) = \sum_{y \sim x} (f(x) - f(y) )$, where $y \sim x$ indicates that the vertices $x$ and $y$ form an edge.  We note that this sign convention is consistent with the one followed in \cites{DodMat06, Woj09} but opposite of \cite{Gri99}.  

In both cases, the Laplacian is essentially self-adjoint and one uses the functional calculus to define a semigroup of operators $e^{-t\Delta}$.  The action of the semigroup is given by a kernel $p_t(x,y)$, henceforth called the \emph{heat kernel}, so that, for $t>0$ and $x \in M$
\[ e^{-t \Delta} u_0(x) = \int_M p_t(x,y) u_0(y) dy \]
where $u_0$ is a continuous, bounded function.  

Alternatively, the heat kernel can be constructed by an exhaustion argument.  In this construction, one takes a sequence of increasing subsets of the whole space, defines the heat kernel with Dirichlet boundary conditions for each subset in the exhaustion, and then passes to the limit.  This is the approach taken in the case of manifolds in \cite{Dod83} and for graphs in \cite{Woj07, Web08} and the equivalence of the two constructions is demonstrated there.  

The heat kernel is positive, symmetric, satisfies the semigroup property and has total integral less than or equal to 1.  In particular, $p_t(x,y)$ gives the transition density for a diffusion process on the underlying space, referred to as the \emph{minimal diffusion process associated to the Laplacian} (see \cite{Az74, Chu60, Gri99}).

\begin{defn}
The minimal diffusion process associated to the Laplacian is \emph{stochastically incomplete} if
\[ \int_M p_t(x,y) dy < 1 \]
for some (equivalently, all) $t>0$ and $x\in M$.
\end{defn} 
\noindent We will follow convention and say that the space is stochastically incomplete if this is the case.  Note that, another way of writing this is $e^{-t \Delta} \mathbf{1} < \mathbf{1}$ where $\mathbf{1}$ indicates the constant function whose value is 1. 

The following theorem gives some equivalent formulations of stochastic incompleteness.  In particular, stochastic incompleteness is equivalent to the non-uniqueness of bounded solutions of the heat equation.  These criteria originate in the works of  R.Z. Has{\cprime}minski{\u\i} \cite{Has60} and W. Feller \cite{Fel54} who studied the question for one-dimensional diffusions.  For a full historical overview and proof in the case of manifolds see \cite{Gri99}*{Theorem 6.2}, for the case of graphs \cite{Woj09}*{Theorem 3.1}, for more general operators which arise as generators of regular Dirichlet forms on discrete sets with an arbitrary measure of full support \cite{KelLen09}*{Theorem 1}.

\begin{thm} \label{stochasticinc}
The following statements are equivalent:
\begin{enumerate}
\item[(1)]  $\int_M p_t(x,y) dy < 1$ for some (equivalently, all) $t>0$ and $x\in M$.   
\item[(2)]  For every $\lambda > 0$, there exists a bounded, positive function $v$ satisfying  \\ $(\Delta + \lambda) v = 0.$
\item[(3)]  For every $\lambda > 0$, there exists a bounded, non-negative, non-zero function $v$ satisfying $(\Delta + \lambda) v \leq 0.$
\item[(4)]  There exists a non-zero, bounded function $u: M \times (0,\infty) \to \mathbb{R}$ satisfying 
\[ \left\{  
\begin{array}{ll} \Delta u(x,t) + \pdtu(x,t) = 0 & \textrm{ for } x \in M, \ t > 0\\
 \lim_{t \to 0^+} u(\cdot, t) \equiv 0. 
 \end{array} \right.  \]
\end{enumerate}
\end{thm}

\begin{rem} \label{massiverem}
The condition (3) is sometimes referred to by saying that $M$ is $\lambda$-\emph{massive}.  More generally, an open subset $\Omega$ of  $M$ is called $\lambda$-\emph{massive}, if there exists a bounded, non-negative function $v$ satisfying $(\Delta + \lambda)v \leq 0$ on $M$ such that $v_{| M \setminus \Omega} \equiv 0$ and v is non-zero on $\Omega$.  By a maximum principle argument, it is easy to see that this property is preserved by enlarging $\Omega$ or by removing a compact subset from $\Omega$ \cite{Gri99}*{Proposition 6.1}.  Therefore, $M$ is stochastically incomplete if it contains a $\lambda$-massive subset and, furthermore, stochastic incompleteness is preserved under the operation of removing a compact subset from $M$.
\end{rem}

\begin{rem}
For another formulation of stochastic completeness in terms of a weak form of the Omori-Yau maximum principle see \cite{PigRigSet03}*{Theorem 1.1}.
\end{rem}

\section{Stochastically Incomplete Manifolds}
In this section we survey some of the known examples of stochastically incomplete manifolds.  As mentioned in the introduction, the discovery of such examples goes back to the work of Azencott.  In particular, letting $K(r)$ and $k(r)$ denote the maximal and minimal absolute values of the sectional curvatures at distance $r$ from a fixed point on a negatively curved, simply connected, analytic manifold $M$, Azencott \cite{Az74}*{Proposition 7.9} proved that: 
\begin{itemize}
\item[(i)] If $\frac{1}{r}\int_0^r K(s) ds \leq C$ for $r$ large, then $M$ is stochastically complete.
\item[(ii)] If $k(r) \geq C r^{2+ \epsilon}$ for $\epsilon>0$, then M is stochastically incomplete. 
\end{itemize}
These results are achieved by applying criteria for the explosion time of the diffusion in terms of the coefficients of the operator.  

In \cite{Yau78}, the heat kernel is analyzed and, in particular, it is shown that, if the Ricci curvature of $M$ is bounded from below, then $\int_M p_t(x,y) dy = 1$.  This was reproven by using a maximum principle argument to show uniqueness of bounded solutions of the heat equation under the same assumption in \cite{Dod83}*{Theorem 4.2}.  Therefore, any manifold whose Ricci curvature is bounded from below is stochastically complete.  This was extended by P. Hsu who showed that, if $\kappa(r)$ denotes any function satisfying $\kappa^2(r) \geq -\inf_{x \in B_r} \textup{Ric}(x)$, then $\int^\infty \frac{1}{\kappa(r)} dr = \infty$ implies stochastic completeness \cite{Hsu89}.   Related results were also given by K. Ichihara \cite{Ich82}*{Theorem 2.1} by comparing with the case of model manifolds, and M. Murata \cite{Mur95}*{Theorem A} who studied the uniqueness of non-negative solutions of the heat equation.   

A different type of criterion for stochastic completeness was given by M.P. Gaffney in \cite{Gaf59}.  Letting $r(x)$ denote the distance to a fixed reference point, Gaffney proves that $M$ is stochastically complete if $e^{-cr(\cdot)}$ is integrable on $M$ for all positive constants $c$.  This gives rise to a volume criterion for stochastic completeness.  If $V(r)$ denotes the Riemannian volume of a geodesic ball in $M$, then L. Karp and P. Li, in an unpublished article, showed that $V(r) \leq e^{cr^2}$ implies stochastic completeness by studying solutions of the heat equation \cite{KarLi}.  A better volume growth condition given by Grigor{\cprime}yan states that if 
\begin{equation} \label{Gri}
\int^\infty \frac{r}{\textup{log } V(r)} dr =  \infty, 
\end{equation}
then $M$ is stochastically complete \cite{Gri86}, see also \cite{Gri99}*{Theorem 9.1}.  This criterion was extended to the more general setting of local Dirichlet spaces by K.T. Sturm \cite{St94}*{Theorem 4}.

That Grigor{\cprime}yan's criterion (\ref{Gri}) is sharp can be seen by considering the case of \emph{spherically symmetric} or \emph{model} manifolds $M_\sigma$.  These are manifolds, diffeomorphic to  $\mathbb{R}^n = \mathbb{R}_+ \times S^{n-1}$, which, following the removal of some number of points, have well-defined polar coordinates $(r, \theta_1,\ldots , \theta_{n-1})$, and whose Riemannian metric is given by $g = dr^2 + \sigma^2(r) g_{S^{n-1}}$.   Here, $g_{S^{n-1}}$ denotes the standard Euclidean metric on $S^{n-1}$ and $\sigma$ is a smooth function satisfying $\sigma(0)=0$ and $\sigma^\prime(0) =1.$   In particular, the area of a geodesic sphere is given by $S(r) =  \omega_n \sigma^{n-1}(r)$ where $\omega_n$ is the area of the sphere in $\mathbb{R}^n$.  See \cite{Gri99}*{Section 3} or \cite{GreWu79} for details.  It can be shown \cite{Gri99}*{Corollary 6.8} that model manifolds are stochastically complete if and only if
\begin{equation}\label{model}
\int^\infty \frac{V(r)}{S(r)} dr = \infty.
\end{equation} 
In particular, if one chooses $\sigma(r)$ so that $V(r) \geq e^{r^{2+\epsilon}}$ for $\epsilon>0$, then $M_\sigma$ is stochastically incomplete.

\begin{rem}
Grigor{\cprime}yan asked \cite{Gri99}*{Problem 9 on page 238} if the condition (\ref{model}) could replace (\ref{Gri}) in implying stochastic completeness for general manifolds.  In a recent paper, C. B\"ar and G.P. Bessa give examples of connected sums of model manifolds which satisfy (\ref{model}) but are stochastically incomplete \cite{BarBes09}*{Theorem 1.3}.
\end{rem}

\section{Stochastically Incomplete Graphs} 
We now begin our discussion of stochastically incomplete graphs.  We consider $G = (V,E)$ an infinite, locally finite, connected graph with vertex set $V$ and edge set $E \subset V \times V$.  We do not consider the case of multiple edges or loops.  We use the notation $x \sim y$ to indicate that the vertices $x$ and $y$ form an edge and say that $x$ and $y$ are \emph{adjacent} or \emph{neighbors} if this is the case.  We let $m(x) = | \{ y \ | \ y \sim x \}|$ denote the \emph{valence} or \emph{degree} of $x$, that is, the number of neighbors of $x$.  

We equip the graph with the usual metric, that is, $d(x,y)$, the \emph{distance} between the vertices $x$ and $y$, is defined as the number of edges in the shortest path connecting $x$ and $y$.  We then fix a vertex $x_0$ and let $S_r = S_r(x_0) = \{ x \ | \ d(x,x_0) = r \}$ denote the \emph{sphere of radius $r$} about $x_0$.  We let $S(r)$ denote the \emph{area} of a sphere of radius $r$, defined as the number of vertices in $S_r$, and $V(r) = \sum_{i=0}^r S(i)$ denote the \emph{volume} of a ball of radius $r$, in analogy with the case of Riemannian manifolds.   

We consider real-valued functions on the vertices of $G$ and, if $f$ is such a function, we define the Laplacian by
\[ \Delta f(x) = \sum_{y \sim x} \big( f(x) - f(y) \big) = m(x) f(x) - \sum_{y \sim x} f(y). \]
As mentioned previously, $\Delta$ is positive and essentially self-adjoint  on the space of finitely supported functions on $V$, which is a dense subset of the Hilbert space of square summable functions on $V$.  Furthermore, the Laplacian is a bounded operator if and only if $m(x) \leq K$ for all vertices $x$.  See, for example, \cites{Dod84, DodMat06, Kel09, Web08, Woj07} for some of the basic properties of the Laplacian.

\begin{rem} \label{renormalized Laplacian remark}
Many authors (e.g. \cite{MedSet00}) consider a different Laplacian which acts on functions on vertices by $\widehat{\Delta} f(x) = f(x) - \frac{1}{m(x)} \sum_{y \sim x} f(y)$ and is related to our Laplacian through $\widehat{\Delta} = \frac{1}{m(\cdot)}\Delta$.  The operator $\widehat{\Delta}$ is bounded on the space of square summable functions, with a weighted inner product, and it can be shown that $e^{-t \widehat{\Delta}} \mathbf{1} = \mathbf{1}$ for all $t>0$, where $\mathbf{1}$ denotes the function which is 1 on all vertices of $G$.  In particular, the minimal diffusion process associated to $\widehat{\Delta}$ is always stochastically complete \cite{Woj07}.  See \cite{Kel09} for other differences between $\Delta$ and $\widehat{\Delta}$.
\end{rem} 

In \cite{DodMat06}, Dodziuk and Mathai study the heat equation using the maximum principle as developed in \cite{Dod83}.  They show that, if $m(x) \leq K$ for all vertices $x$, bounded solutions of the heat equation on $G$ are uniquely determined by initial data.  In particular, all such graphs are stochastically complete.  

This result was extended in two ways.  In \cite{Dod06}, Dodziuk applied this technique to a Laplacian with weights.  More specifically, consider a weighted graph, that is, a graph where each edge $x \sim y$ is assigned a positive, symmetric weight $a_{x, y}$.  The resulting weighted Laplacian $A$ acts on functions by $Af(x) = \sum_{y \sim x} a_{x,y} \big( f(x) - f(y) \big).$  In \cite{Dod06}*{Theorem 4.1}, it is shown that, if $m(x) \leq K_1$ for all vertices $x$ and $a_{x, y} \leq K_2$ for all edges $x \sim y$, then the heat equation involving this Laplacian has unique bounded solutions and, as such, $e^{-t A} \mathbf{1} = \mathbf{1}$ for all $t \geq 0$.  

In \cite{Web08}, A. Weber replaced the assumption $m(x) \leq K$ with a different curvature condition on the graph.  Specifically, letting, $r(x) = d(x, x_0)$ where $x_0$ is a fixed reference vertex, if follows that $\Delta r(x) \geq -C$ for $C \geq 0$ implies that the graph is stochastically complete \cite{Web08}*{Corollary 4.15}.  To give a geometric interpretation to this last result, let $m_{\pm}(x) = | \{ y  \ | \ y \sim x \textup{ and }   r(y) = r(x) \pm 1 \} |$ denote the number of vertices one step further and closer, respectively, from $x_0$ then is $x$.   It follows that $\Delta r(x) \geq -C$ if and only if $m_+(x) - m_-(x) \leq C.$  Hence, in particular, a graph will be stochastically complete if it is not expanding too much, relative to the number of incoming edges, in all directions. 

 All these results were obtained by applying a maximum principle to study bounded solutions of the heat equation.  Furthermore, the geometric assumptions are imposed at each vertex of the graph.  However, the stochastic completeness or incompleteness of a graph should be determined by geometric conditions at infinity as in the Riemannian setting.

In \cites{Woj07, Woj09} a different technique is applied.  Specifically, the same question is approached through the study of positive solutions to the difference equation $(\Delta + \lambda) v(x) = 0$ for a positive constant $\lambda$.  By (2) in Theorem \ref{stochasticinc},  stochastic incompleteness is equivalent to the boundedness of the solution $v$.  This fact was used to obtain a general criterion for stochastic completeness extending the result of Dodziuk and Mathai.   Specifically, fixing a vertex $x_0$ and letting $K(r) = \max_{x \in S_r(x_0)} m(x)$, Theorem 3.2 in \cite{Woj09} states that, if $\sum_{r=0}^\infty \frac{1}{K(r)} = \infty$, then $G$ is stochastically complete.  We now sharpen this result as follows: we let $m_+(x) = | \{ y \ | \ y \sim x \textup{ and } r(y) = r(x) + 1 \} |$ as above, and let $K_+(r) = \max_{x \in S_r} m_+(x)$.
\begin{thm} \label{completenesstheorem}
If 
\[ \sum_{r=0}^{\infty} \frac{1}{K_+(r)} = \infty, \]
then $G$ is stochastically complete.
\end{thm}

\begin{proof}
Let $v>0$ satisfy $(\Delta + \lambda) v(x) = 0$ for $\lambda > 0$ and $x \in G$.  Let $w(r) = \max_{x \in S_r} v(x)$.  Then, $(\Delta + \lambda) v(x_0) =0$ implies that $\sum_{y \sim x_0} \big( v(y) - v(x_0) \big) = \lambda v(x_0)$.  Therefore, as $w(0) = v(x_0)$,
\[ K_+(0) \big( w(1) - w(0) \big) \geq \sum_{y \sim x_0} \big( v(y) - w(0) \big) = \lambda w(0) \]
so that 
\begin{equation}\label{firstjump}
w(1) - w(0) \geq \frac{\lambda}{K_+(0)}  w(0). 
\end{equation}

Now, choose $x_r \in S_r$ such that $w(r) = v(x_r)$.  Then, $(\Delta + \lambda) v(x_r) = 0$ implies that
\begin{equation} \label{nextjumps}
 \sum_{\substack{y \sim x_r \\ y \in S_{r+1}} }\big( v(y) - v(x_r) \big) = \sum_{\substack{ y \sim x_r \\ y \not \in S_{r+1}}} \big( v(x_r) - v(y) \big) + \lambda v(x_r). 
\end{equation}
Using (\ref{firstjump}) and (\ref{nextjumps}), it follows by induction that $w(r+1) - w(r) > 0$ for all $r \geq 0$.  Therefore,
\begin{align*}
K_+(r) \big( w(r+1) - w(r) \big) & \geq  \sum_{\substack{y \sim x_r \\ y \in S_{r+1}} }\big( v(y) - w(r) \big)   \\
& = \sum_{\substack{ y \sim x \\ y \not \in S_{r+1}} } \big( w(r) - v(y) \big) + \lambda w(r) > \lambda w(r). 
\end{align*}
This implies that
\[ w(r+1) - w(r) > \frac{\lambda}{K_+(r)}  w(r) > \frac{\lambda}{K_+(r)} w(0). \]
Therefore, $\sum_{r=0}^\infty  \frac{1}{K_+(r)} = \infty$ implies $\sum_{r=0}^\infty \big( w(r+1) - w(r) \big) = \infty$ so that $w$, and, therefore, $v$, is unbounded.
\end{proof}

\begin{rem}
For trees, we have the following complementary result: if $k_+(r) = \min_{x \in S_r} m_+(x) > 0$, then the tree is stochastically incomplete if $\sum_{r=0}^\infty \frac{1}{k_+(r)} < \infty$ \cite{Woj09}*{Theorem 3.4}.  In particular, in the case of \emph{spherically symmetric trees} $T_k$, that is, when $k(r) = k_+(r) = K_+(r)$ denotes the branching number of $T_k$, we have that $T_k$ is stochastically incomplete if and only if 
\[ \sum_{r=0}^\infty \frac{1}{k(r)} < \infty. \]
\end{rem}

\subsection{Stochastically Incomplete Subgraphs} \label{subgraphs section}
As mentioned in Remark \ref{massiverem}, $\lambda$-massiveness is preserved under the operation of removing a compact (or finite, in the graph case) subset.  In particular, this can be used to show that a graph which contains a stochastically incomplete subgraph with only finitely many vertices adjacent to vertices that are not in the subgraph is stochastically incomplete.  That is, if $G \subset \widetilde{G}$ with $G$ stochastically incomplete and $\partial G = \{ x \in G \ | \ \exists \ y \sim x \textup{ with } y \not \in G \}$ a finite set, then $\widetilde{G}$ is stochastically incomplete.  

This result was extended by M. Keller and D. Lenz in \cite{KelLen09}*{Theorem 3} where they give a general condition for stochastic incompleteness of $\widetilde{G}$ in terms of $G$ by considering the Dirichlet Laplacian on $G \subset \widetilde{G}$.  In fact, it should be pointed out, that Keller and Lenz consider processes associated to much more general operators of the form $L = B + C$, where $B$ is a weighted Laplacian and $C$ is a potential,  with $L$ acting on an $\ell^2$ space with an arbitrary measure of full support; furthermore,  their graphs are not necessarily locally finite, but we specialize their results to our setting.  In particular, they show that, if $G \subset \widetilde{G}$ is a stochastically incomplete subgraph and $\sup_{x \in \partial G} | \{ y \sim x \ | \ y \not \in G \} | < \infty$, then $\widetilde{G}$ is stochastically incomplete.

Moreover, they show that every stochastically incomplete graph $G$ is a subgraph of a stochastically complete graph $\widetilde{G}$ \cite{KelLen09}*{Theorem 2}.  They construct the supergraph $\widetilde{G}$ by attaching, to each vertex $x$ in $G$, $m(x)d(x)$ copies of the graph $\mathbb{W}$ with vertex set $\{0, 1, 2, \ldots \}$ and edge set $\{ (n, n+1) \ | \ n=0,1,2, \ldots \}$.  Here, $m(x)$ denotes the valence of the vertex $x$ in $G$ and attaching means that the vertex $x$ is associated to the vertex 0 in $\mathbb{W}$.  They show that, if $d(x)$ is chosen so that $\sum_{r=0}^\infty \frac{1}{d(x_r)} = \infty$ for every sequence of vertices such that $x_r \sim x_{r+1}$,  then $\widetilde{G}$ is stochastically complete.  

We now show that this result is optimal by analyzing the stochastic completeness of the supergraph constructed this way in the case of spherically symmetric trees.  As mentioned above, for spherically symmetric trees, $k(r) = k_+(r) = K_+(r)$ and they are stochastically incomplete if and only if $\sum_{r=0}^\infty \frac{1}{k(r)} < \infty.$  Such trees are determined by the function $k$ and we denote them by $T_k$.  Now, instead of attaching the graphs $\mathbb{W}$ as in \cite{KelLen09}, we connect $\tilde{k}(r)$ \emph{end} vertices, that is, vertices of valence 1, to each vertex $x \in S_r \subset T_k$.  As Keller and Lenz point out, this construction has an equivalent effect (see, Remark \ref{constants} below).  We denote the resulting graph $T_k^{\tilde{k}}$.  Therefore, each $x \in T_k \subset T_k^{\tilde{k}}$ in $S_r$ has $k(r)$ neighbors in $S_{r+1} \subset T_k$ and $\tilde{k}(r)$ end vertex neighbors in $T_k^{\tilde{k}} \setminus T_k$.

\begin{thm} \label{doubletrees}
$T_k^{\tilde{k}}$ is stochastically incomplete if and only if
\[ \sum_{r=0}^\infty \frac{\tilde{k}(r) + 1}{k(r)} < \infty. \]
\end{thm}

\begin{proof}
Let $v>0$ satisfy $(\Delta + \lambda)v(x)=0$ for $\lambda > 0$ and $x \in T_k^{\tilde{k}}$.   For every $\tilde{y} \in T_k^{\tilde{k}} \setminus T_k$ there exists a unique $x \sim \tilde{y}$ such that $x \in T_k$.   Then, $(\Delta + \lambda)v(\tilde{y})= \big( v(\tilde{y}) - v(x) \big) + \lambda v(\tilde{y}) = 0$, implies that $v(\tilde{y}) = \left( \frac{1}{1+\lambda} \right) v(x).$  In particular, for every $\tilde{y} \in T_k^{\tilde{k}} \setminus T_k$ with $x \sim \tilde{y}$ and $x \in T_k$, we have
\begin{equation}\label{differences}
v(x) - v(\tilde{y}) = \left( \frac{\lambda}{1+ \lambda} \right) v(x) = \alpha v(x)
\end{equation}
for $\alpha = \frac{\lambda}{1+\lambda}$.

Now, by averaging over spheres, it suffices to consider functions $v$ whose value depends only on the distance from $x_0$.  We denote the restriction of $v$ to $T_k$ by $w$.  Then, $(\Delta + \lambda) v(x_0) = 0$ and (\ref{differences}) imply that $k(0) \big( w(0) - w(1) \big) + \big(\alpha \tilde{k}(0) + \lambda \big) w(0) =0.$  Therefore,
\begin{equation} \label{firstincrement}
w(1)-w(0) = \left( \frac{\alpha \tilde{k}(0) + \lambda}{k(0)} \right) w(0). 
\end{equation}
If $x \in S_r \subset T_k$ for $r > 0$, then $(\Delta+\lambda)v(x) = k(r) \big( w(r)-w(r+1) \big) + \big( w(r) - w(r-1) \big) + (\alpha \tilde{k}(r) + \lambda) w(r) = 0.$  Therefore,
\begin{equation}\label{increments}
w(r+1) - w(r) = \left( \frac{\alpha \tilde{k}(r) + \lambda}{k(r)} \right) w(r) + \frac{1}{k(r)} \big( w(r) - w(r-1) \big).
\end{equation}

Applying (\ref{firstincrement}) and (\ref{increments}), it follows by induction that $w(r+1) - w(r) > 0$ for all $r\geq0$.  Therefore, the increment $w(r+1)- w(r)$ can be estimated as follows:
\[ \left( \frac{\alpha \tilde{k}(r) + \lambda} {k(r)} \right) w(r) \leq w(r+1) - w(r) < \left( \frac{\alpha \tilde{k}(r) + \lambda + 1}{k(r)} \right) w(r). \]
Therefore, 
\[ \left(1 + \frac{\alpha \tilde{k}(r) + \lambda} {k(r)} \right) w(r) \leq w(r+1)  < \left(1+ \frac{\alpha \tilde{k}(r) + \lambda + 1}{k(r)} \right) w(r) \]
and iterating this down to $r= 0$ gives
\[  \prod_{i=0}^r \left(1 + \frac{\alpha \tilde{k}(i) + \lambda} {k(i)} \right) w(0) \leq w(r+1)  < \prod_{i=0}^r \left(1+ \frac{\alpha \tilde{k}(i) + \lambda + 1}{k(i)} \right) w(0). \]
It follows that $T_k^{\tilde{k}}$ is stochastically incomplete if and only if
 \[ v \textup{ is bounded } \Leftrightarrow  w \textup{ is bounded } \Leftrightarrow \prod_{r=0}^\infty \left( 1 + \frac{\tilde{k}(r) + \lambda} {k(r)} \right) < \infty \Leftrightarrow \sum_{r=0}^\infty \frac{\tilde{k}(r) + 1}{k(r)} < \infty. \]
\end{proof}

\begin{ex} \label{completionex}
As an example, we let $k(r)=(r+1)^2$ and $\tilde{k}(r) = r+1$ so that $T_k$ is stochastically incomplete and $T_k \subset T_k^{\tilde{k}}$ with $T_k^{\tilde{k}}$ complete.  Note that, on $T_k^{\tilde{k}}$, $K_+(r) = (r+1)^2 + (r+1)$ while $k_+(r)=0$ for $r>0$, so that neither of our general results concerning the stochastic completeness of trees apply in this case.
\end{ex}

\begin{rem}
If $\tilde{k}(r)=0$ for all $r \geq 0$, then we recover the result for spherically symmetric trees mentioned previously, that is, $T_k$ is stochastically incomplete if and only if $\sum_{r=0}^\infty \frac{1}{k(r)} < \infty.$
\end{rem}

\begin{rem}\label{constants}
The only difference between this construction, where we connect end vertices to each vertex in $T_k$, and the one given in \cite{KelLen09}, where one attaches a path to infinity, is in the constant $\alpha = \frac{\lambda}{1+ \lambda}$ in (\ref{differences}).  Namely, when one attaches a path to a vertex $x \in T_k$, by identifying that vertex with the vertex 0 in the graph $\mathbb{W}$, then a calculation with difference equations gives that 
\[ v(x) - v(1) = \frac{ \lambda + \sqrt{\lambda(\lambda + 2)} } {\lambda + 2 +  \sqrt{\lambda(\lambda+2)} } v(x) = \beta v(x) \]
and one replaces the constant $\alpha$ in (\ref{differences}) with the constant $\beta$.  We mention this fact to make it clear, as Keller and Lenz do, that it is not necessary to add end vertices to a graph in order to construct the complete supergraph.
\end{rem}

\subsection{Spherically Symmetric Graphs}
We now give a necessary and sufficient condition for the stochastic incompleteness of spherically symmetric graphs and illustrate the result with several examples.  We consider graphs with a vertex $x_0$ such that $m_{\pm}(x) = | \{ y \ | \ y \sim x \textup{ and } r(y) = r(x) \pm 1 \} |$ where $r(y) = d(y,x_0)$, depend only on the distance between $x$ and $x_0$.  We write $k_{\pm}(r)$ for the common values of $m_{\pm}(x)$ on $S_r$.  We call such graphs \emph{spherically symmetric} and denote them by $G_{k_\pm}$.   For the special case of spherically symmetric trees $T_k$ discussed previously, it follows that $k_+(r) = k(r)$ and $k_-(r) =1$.  

Note, that we make no assumptions concerning $m_0(x) = | \{ y \ | \ y \sim x  \textup{ and }  r(y) = r(x) \} |$ and it will become clear in the course of the proof of the next theorem that these quantities play no role in the stochastic incompleteness of $G_{k_\pm}$.  This is somewhat surprising as one might expect that by choosing $m_0(x)$ large and making the spheres highly connected one might slow the diffusion down, but this is not the case and we explain this below.   We recall that $S(r)$ denotes the number of vertices in the sphere of radius $r$ about $x_0$, while $V(r)$ denotes the number of vertices in the ball of radius $r$, so that, $V(r) = \sum_{i=0}^r S(i).$ 

\begin{thm}\label{sgraphs}
$G_{k_\pm}$ is stochastically incomplete if and only if
\[ \sum_{r=0}^\infty \frac{V(r)}{k_+(r)S(r)} < \infty. \]
\end{thm}
\begin{proof}
Let $v>0$ satisfy $(\Delta + \lambda)v(x)=0$ for $\lambda>0$ and $x \in G_{k_{\pm}}$.
By averaging over spheres, as above, it suffices to consider functions depending only of the distance from $x_0$.  That is, if $v$ satisfies the conditions above, then, using the identities $k_+(r)S(r) = k_-(r+1)S(r+1)$, it follows that $w(r) = \frac{1}{S(r)} \sum_{x \in S_r} v(x)$ satisfies $(\Delta + \lambda)w(r) = 0$ for all $r \geq 0$.  Note that, it is at this point that the terms involving $m_0(x)$ cancel out.  

Therefore, we only consider positive functions $w$ satisfying $(\Delta + \lambda) w(r) = 0$ for $\lambda>0$ and $r \geq 0$.  We claim that 
\begin{equation}\label{sincrements}
w(r+1) - w(r) = \frac{\lambda}{k_+(r) S(r)} \sum_{i=0}^r S(i) w(i).
\end{equation}
This follows by induction.  For $r=0$, we have $(\Delta + \lambda) w(0) = k_+(0)\big( w(0) - w(1) \big) + \lambda w(0) =0$ which implies that
\[ w(1) - w(0) = \frac{\lambda}{k_+(0)} w(0). \]

Now, for $r>0$, $(\Delta + \lambda) w(r) =0$ implies that $k_+(r)\big( w(r+1) - w(r) \big) = k_-(r) \big( w(r) - w(r-1) \big) + \lambda w(r) $.  Therefore, if (\ref{sincrements}) holds for $w(r) - w(r-1)$, then, applying $k_+(r-1)S(r-1) = k_-(r)S(r)$, we obtain
\begin{align*}
w(r+1) - w(r) &= \frac{k_-(r)}{k_+(r)} \big( w(r) - w(r-1) \big) + \frac{\lambda}{k_+(r)} w(r) \\
& = \frac{k_-(r)}{k_+(r)} \left( \frac{\lambda}{k_+(r-1) S(r-1)} \sum_{i=0}^{r-1} S(i) w(i) \right) + \frac{\lambda}{k_+(r)} w(r) \\
&= \frac{\lambda}{k_+(r) S(r)} \sum_{i=0}^{r-1} S(i) w(i) + \frac{\lambda}{k_+(r)} w(r) \\
& = \frac{\lambda}{k_+(r) S(r)} \sum_{i=0}^r S(i)w(i)
\end{align*}
establishing (\ref{sincrements}).

In particular, (\ref{sincrements}) implies that $w(r+1) > w(r)$ for all $r \geq 0$ so that the increments $w(r+1) - w(r)$ can be estimated as follows:
\[ \frac{\lambda V(r)}{k_+(r) S(r)} w(0) \leq w(r+1) - w(r) \leq \frac{\lambda V(r)}{k_+(r)S(r)} w(r). \]
Therefore, if $\sum_{r=0}^\infty \frac{V(r)}{k_+(r) S(r)} = \infty$, then $\sum_{r=0}^\infty \big( w(r+1) - w(r) \big) = \infty$ and $w$ is unbounded.  On the other hand,
\[ w(r+1) \leq \left( 1 + \frac{\lambda V(r)}{k_+(r) S(r)} \right) w(r) \leq \prod_{i=0}^r \left( 1 + \frac{\lambda V(i)}{k_+(i)S(i)} \right) w(0)  \]
so that, if $\sum_{r=0}^\infty \frac{V(r)}{k_+(r)S(r)} < \infty$, then $w$ is bounded.
\end{proof}

We now illustrate the theorem by giving a series of examples.
\begin{ex} \label{modelgraphsex}
For the case of spherically symmetric trees, $k_+(r) = k(r)$, the branching number, and $k_-(r) = 1$.  Furthermore, $k_+(r)S(r) = S(r+1)$, so that Theorem \ref{sgraphs} implies that $T_k$ is stochastically incomplete if and only if
\begin{equation}\label{strees1}
 \sum_{r=0}^\infty \frac{V(r)}{S(r+1)} < \infty. 
\end{equation}
For such trees,  $S(r) = \prod_{i=0}^{r-1} k_+(i)$ and, by the limit comparison test for series, (\ref{strees1}) is equivalent to 
\begin{equation}\label{strees2}
\sum_{r=0}^\infty \frac{1}{k_+(r)} < \infty 
\end{equation}
which was obtained as Corollary 3.8 in \cite{Woj09} and, as a special case, in Theorem \ref{doubletrees}.

One can extend these examples by connecting any number of vertices on a particular sphere $S_r$ in $T_k$ and this has no effect on the stochastic completeness.  That is, let $T_k \subset G_k$, where $G_k$ is obtained by connecting some number of vertices on each sphere $S_r$.  Then, Theorem \ref{sgraphs} shows that $G_k$ is stochastically incomplete if and only if 
\begin{equation} \label{modelgraphs}
\sum_{r=0}^\infty \frac{V(r)}{S(r+1)} < \infty. 
\end{equation}

Furthermore, by applying heat kernel comparison theorems, it follows that the heat kernels on $T_k$ and $G_k$ are equal.  Specifically, it is easy to see that the heat kernel on $T_k$ depends only on the distance from $x_0$, that is, $p_t^{T_k}(x_0, x) = p_t^{T_k}(r)$ for all $x \in S_r$.  Then, applying Theorem 3.11 in \cite{Woj09}, it follows that $p_t^{G_k} (x_0, x) = p_t^{T_k}(r)$ for all $x \in S_r$, where $p_t^{G_k}(x_0,x)$ denotes the heat kernel on $G_k$.  This gives another proof of the fact that $G_k$ is stochastically incomplete if and only if $T_k$ is, since $G_k$ and $T_k$ have the same set of vertices.

We believe that the graphs $G_k$ are the analogues of model manifolds and the criterion (\ref{modelgraphs}) corresponds to the condition $\int^\infty \frac{V(r)}{S(r)} dr < \infty$ in (\ref{model}).  Specifically, $dV(r) := V(r+1) - V(r) = S(r+1)$ plays the role of $V^\prime(r) = S(r)$ from the manifold case.  The essential point for the correspondance, apart from the spherical symmetry, is that each vertex in $G_k$ has a unique shortest path connecting it to the origin vertex $x_0$, which is also the case for model manifolds.

\begin{rem}  It is surprising that the criteria (\ref{strees1}) and (\ref{strees2}) for stochastic incompleteness apply to $G_k$ as well as to $T_k$.  For example, take $T_k$ with $k_+(r) = k(r) = (r+1)^2$ so that $T_k$ is stochastically incomplete with $S(r) = (r!) ^2$.  Now, connect each vertex $x \in S_r$ to every other vertex in $S_r$ to obtain $G_k$.  Then, at $x \in S_r \subset G_k$, $m_+(x) = k_+(r) = (r+1)^2$, while $m_0(x) = ( r! )^2 - 1$ so, probabilistically, the particle is much more likely to remain on the sphere $S_r$ than go outwards.  On the other hand, Theorem \ref{doubletrees} shows that adding only $\tilde{k}(r)=r+1$ end vertices at $x$ does have the effect of trapping the particle.

This contrast can be understood in light of the following model for the diffusion process governed by the heat kernel.  The direction of each jump is chosen randomly with probabilities as in the case of the simple random walk on the graph; however, the holding time of the particle at a vertex is a random variable whose exponential distribution depends on the valence of the vertex.  Specifically, if the particle undergoing the diffusion is at a vertex $x$ at time $s$, then it will jump, after a random time, to one of the neighbors of $x$ with probability $\frac{1}{m(x)}$.  Furthermore, the probability that it has not moved from the vertex $x$ at time $s+t$ is $e^{-t m(x)}$.  See \cites{Chu60} for more details on this model of the diffusion process.

Therefore, in the example above, although the particle will most likely jump to another vertex on the same sphere, the factorial valence at that vertex implies that the particle will not remain there for long.  In the case of adding end vertices, the holding time of the particle at a vertex of valence one is expected to be longer then at a vertex of high valency and this explains why, in this case, the particle is slowed down and explosion is prevented.  \end{rem}
\end{ex}

\begin{ex} \label{sphericalgraphs}
Let $S(r)$ be given with $S(0) = 1$.  We then connect every vertex in $S_r$ to every vertex in $S_{r+1}$ for all $r \geq 0$.  Such graphs are spherically symmetric with $k_+(r) = S(r+1)$ and $k_-(r) = S(r-1)$ and we denote them by $G_S$.  Theorem \ref{sgraphs} then implies that $G_S$ is stochastically incomplete if and only if
\begin{equation}\label{spheregraphs}
\sum_{r=0}^\infty \frac{\sum_{i=0}^r S(i)}{S(r+1) S(r)} < \infty.
\end{equation}

This allows us to create many examples of stochastically incomplete graphs with polynomial volume growth.  For example, letting $S(r) = (r+1)^3$ then, by (\ref{spheregraphs}), $G_S$ is stochastically incomplete with $V(r) = \frac{(r+1)^2(r+2)^2}{4}.$  Moreover, these examples show that the condition $\sum_{r=0}^\infty \frac{1}{k_+(r)} < \infty$ which is sufficient for the stochastic incompleteness of trees is not, in general, sufficient by considering the case of $k_+(r) = S(r+1) = (r+2)^2$ which is complete.

Furthermore, given any graph, not necessarily spherically symmetric, whose spheres satisfy (\ref{spheregraphs}) for some vertex $x_0$, one can create a spherically symmetric stochastically incomplete graph by connecting all vertices in $S_r$ to all vertices in $S_{r+1}$ for $r\geq 0$.  This operation, where we add edges but do not change the vertex set to obtain a stochastically incomplete graph from a complete one, is complementary to the procedure described in Section \ref{subgraphs section}, where we added vertices and edges to create a complete graph from an incomplete one.  On the other hand, in Example \ref{modelgraphsex} we discussed how adding edges along the same sphere has no effect on stochastic completeness.
\end{ex}

\begin{ex}
One can also extend our techniques to the case of the weighted Laplacian as found in \cite{Dod06} and, as a special case, in \cite{KelLen09}.  For example, consider the weighted graph $\mathbb{W}_a$ with vertex set $V = \{ 0, 1, 2, \ldots \}$, edge set $\{ (r,r+1) \ | \ r = 0,1,2, \ldots \}$, and edge weights $a(r) = a_{r, r+1}$.  Consider a function $v$ on the vertices of $\mathbb{W}_a$ satisfying $(A + \lambda)v(r)= 0$ for $r=0,1,2, \ldots$ where $A$ denotes the weighted Laplacian.  It follows that $v(1) - v(0) = \frac{\lambda}{a(0)} v(0)$ and one shows by induction that
\[ v(r+1)-v(r) = \frac{\lambda}{a(r)} \sum_{i=0}^r v(i). \]
Therefore, estimating as in Theorem \ref{sgraphs}, we have that $v$ will be bounded if and only if
\[ \sum_{r=0}^\infty \frac{r}{a(r)} < \infty \]
which is a special case of Theorem \ref{sgraphs} where $S(r)=1$ and $k_+(r) = a(r).$
\end{ex}

\section{Concluding Remarks}
\subsection{Volume Growth} 
We have shown that, with respect to volume growth, stochastically incomplete manifolds and graphs exhibit quite different behavior.  However, we still believe that a general criterion for stochastic completeness of graphs in terms of volume growth, in analogy to Grigor{\cprime}yan's result, should exist.  On the other hand, as mentioned previously, in \cite{BarBes09}, it is shown that the criterion for stochastic completeness of model manifolds, that is, $\int^\infty \frac{V(r)}{S(r)} dr = \infty$, does not imply stochastic completeness of general manifolds.  Furthermore, it is also shown in \cite{BarBes09} that $\int^\infty \frac{V(r)}{S(r)} dr < \infty$ does not, in general, imply stochastic incompleteness.  We are already in a position to prove the analogous statements for graphs.  

First, we show that the condition $\sum_{r=0}^\infty \frac{V(r)}{S(r+1)} = \infty$ does not imply stochastic completeness.  For this, take $G_S$ with $S(r) = (r+1)^3$ in Example \ref{sphericalgraphs}.  This example is stochastically incomplete but $V(r) > S(r+1)$ for $r$ large.  

To show that $\sum_{r=0}^\infty \frac{V(r)}{S(r+1)} < \infty$ does not imply stochastic incompleteness, consider Example \ref{completionex} where we take $T_k^{\tilde{k}}$ with $k(r)=(r+1)^2$ and $\tilde{k}(r)=r+1$.  By Theorem \ref{doubletrees}, this example is stochastically complete.  Let $S(r)$ and $V(r)$ denote the area of the sphere and volume of the ball in $T_k$ while $\widetilde{S}(r)$ and $\widetilde{V}(r)$ denote the corresponding quantities in $T_k^{\tilde{k}}$.  Then $S(r) = (r!) ^ 2$ and $V(r) = \sum_{i=0}^r S(i)$ and, from Theorem \ref{sgraphs}, we know that $\sum_{r=0}^\infty \frac{V(r)}{S(r+1)} < \infty.$  On, $T_k^{\tilde{k}}$ we have that $\widetilde{S}(r) = S(r) + r S(r-1)$ where we set $S(-1) =0$.  Therefore, $\widetilde{V}(r) = \sum_{i=0}^r \widetilde{S}(i) = V(r) + \sum_{i=0}^r  i S(i-1)$ and we have that
\begin{align*}
\sum_{r=0}^\infty \frac{\widetilde{V}(r)}{\widetilde{S}(r+1)} &= \sum_{r=0}^\infty \frac{V(r) + \sum_{i=0}^r  i S(i-1)}{S(r+1) + (r+1) S(r)} \\
& < \sum_{r=0}^\infty \frac{V(r)}{S(r+1)} +  \sum_{r=1}^\infty \frac{ r V(r-1)}{(r+1)S(r)} < \infty.  
\end{align*}

\subsection{Discretization and Rough Isometries}
There is a well-known discretization procedure originating in the works of M. Kanai in which a graph is associated to a Riemannian manifold in such a way that many properties of the graph reflect those of the manifold \cites{Kan85, Kan86}.  This procedure is used to show that certain properties of manifolds are preserved by rough isometries.  For example, in \cite{Kan86}*{Theorem 1}, Kanai shows that rough isometries preserve the property of a manifold to be parabolic.  The proof uses a graph to approximate the manifold, shows that parabolicity is preserved during the discretization, and then shows that rough isometries preserve the parabolicity of graphs.  The same technique is used in \cite{Kan85} to show that volume growth is preserved by rough isometries.  

However, a crucial assumption in Kanai's construction is that the manifolds and graphs have \emph{bounded geometry}.  For manifolds, this means, in particular, that the Ricci curvature is bounded from below and for graphs this means that the valence is uniformly bounded from above.  As we have mentioned, these assumptions automatically imply that the spaces in question are stochastically complete and Kanai's discretization scheme does not apply in our case.  In particular, in Example \ref{completionex}  the graphs $T_k$ and $T_k^{\tilde{k}}$ are roughly isometric but $T_k$ is stochastically incomplete while $T_k^{\tilde{k}}$ is stochastically complete.

\subsection{Quantum Graphs}
As mentioned previously, Grigor{\cprime}yan's volume criterion (\ref{Gri})  was extended to the more general setting of Dirichlet spaces by Sturm \cite{St94}.  In particular, \emph{metric} or \emph{quantum} graphs are covered by this extension.  Therefore, we have shown that, with respect to volume growth, solutions of the heat equation behave differently on quantum and discrete graphs.  This is also the case for solutions of the wave equation where, in the quantum setting, as on manifolds, solutions of the wave equation have finite propagation speed, in contrast to the case of discrete graphs.  See \cite{FriTil04} for details.


\subsection*{Acknowledgment}
I would like to thank many people with whom I've had helpful and inspiring discussions while working on the material in this article.  In particular, I would like to thank J{\'o}zef Dodziuk, Pedro Freitas, Sebastian Haeseler, Matthias Keller, Daniel Lenz, Erin Pearse, Florian Sobieczky, and Jean-Claude Zambrini.
\end{document}